\documentclass{amsart}
\usepackage{ifpdf}
\usepackage{graphicx,amssymb,lineno}
\usepackage{epsfig}
\usepackage{graphics}
\usepackage[emtex]{color}
\usepackage[usenames,dvipsnames,svgnames,table]{xcolor}
\usepackage{amsfonts}
\usepackage {eepic}
\usepackage{latexsym}
\usepackage{tikz}
\usepackage{graphics}
\usepackage{dsfont}
\usepackage{euscript}
\usepackage{mathrsfs}
\usepackage{graphics}
\usepackage {epic}
\usepackage {eepic}
\usepackage{latexsym}
\usepackage{fancybox}

\usepackage{comment}
\usepackage{setspace}

\usepackage{soul}

\usepackage{enumerate}

\usepackage[algo2e,lined,boxruled]{algorithm2e}
\usepackage{algorithm,algpseudocode}
\usepackage{caption}
\usepackage{subcaption}
\usepackage{graphicx}
\usepackage{epstopdf}
\usepackage[colorlinks]{hyperref}
\hypersetup{colorlinks=true, linkcolor=black,citecolor=black}

 \newtheorem{theorem}{Theorem}[section]
 \newtheorem{lemma}[theorem]{Lemma}
\newtheorem{proposition}[theorem]{Proposition}
\newtheorem{claim}[theorem]{Claim}

\usepackage[english]{babel}
\usepackage[utf8]{inputenc}

\newcommand{\NP}{{\mathsf{NP}}}
\newcommand{\UGC}{{\mathsf{UGC}}}

\newcommand{\APX}{{\mathsf{APX}}}

\addto\extrasenglish{%
}

\begin{document}
\title{On minimum $t$-claw deletion in split graphs}
\author{Sounaka Mishra}
\address{Department of Mathematics, IIT Madras, Chennai-600036, India} 
\email{sounak@iitm.ac.in}
\date{}
\maketitle

\begin{abstract}
 For $t\geq 3$, $K_{1, t}$ is called $t$-claw. In minimum $t$-claw deletion problem (\texttt{Min-$t$-Claw-Del}), given a graph $G=(V, E)$, it is required to find a vertex set $S$ of minimum size such that $G[V\setminus S]$ is $t$-claw free. 
 In a split graph, the vertex set is partitioned into two sets such that one forms a clique and the other forms an independent set. Every $t$-claw in a split graph has a center vertex in the clique partition. This observation motivates us to consider the minimum one-sided bipartite $t$-claw deletion problem (\texttt{Min-$t$-OSBCD}). Given a bipartite graph $G=(A \cup B, E)$, in \texttt{Min-$t$-OSBCD} it is asked to find a vertex set $S$ of minimum size such that $G[V \setminus S]$ has no $t$-claw with the center vertex in $A$. A primal-dual algorithm approximates \texttt{Min-$t$-OSBCD} within a factor of $t$. We prove that it is $\UGC$-hard to approximate with a factor better than $t$. We also prove it is approximable within a factor of 2 for dense bipartite graphs. By using these results on \texttt{Min-$t$-OSBCD}, we prove that \texttt{Min-$t$-Claw-Del} is $\UGC$-hard to approximate within a factor better than $t$, for split graphs. We also consider their complementary maximization problems and prove that they are $\APX$-complete. 
\end{abstract}

\section{Introduction}
Given a graph $G=(V, E)$ and a property $\pi$, in minimum node deletion problem it is required to find a vertex set $S$ of minimum size such that $G[V \setminus S]$ satisfies $\pi$. A detailed study of the computational complexity of these kinds of problems is done when the graph property $\pi$ satisfies certain conditions. We need some definitions to explain these results. A graph property $\pi$ is called nontrivial if an infinite family of graphs satisfies $\pi$ and an infinite class of graphs does not satisfy $\pi$. A property $\pi$ is called hereditary if a graph $G$ satisfies $\pi$ then all its vertex-induced subgraphs satisfy $\pi$. Lewis and Yannakakis \cite{lewis1980node} proved that such kind of node deletion problems are $\NP$-complete, provided $\pi$ is verified in polynomial time. Later, Lund and Yannakakis \cite{lund1993approximation} proved that these problems are  $\APX$-hard for nontrivial hereditary graph properties. They also complemented with constant factor approximation algorithms for these node deletion problems provided the graph property $\pi$ has a finite forbidden graphs characterization.

For $t\geq 3$, the complete bipartite graph $K_{1, t}$ is called a $t$-claw, and the vertex with degree $t$ is called its center vertex. $K_{1,3}$ is widely known claw. In this paper, we consider minimum $t$-claw node deletion problem (\texttt{Min-$t$-Claw-Del}) restricted to split graphs. In \texttt{Min-$t$-Claw-Del}, given a graph $G=(V, E)$ it is required to find a vertex set $S$ of minimum size such that $G[V\setminus S]$ is $t$-claw free. By using the result of Lund and Yannakakis \cite{lund1993approximation}, it can be proved that \texttt{Min-$t$-Claw-Del} is $\APX$-complete and can be approximated within a factor of $t+1$. It can be formulated as a $(t+1)$-hitting set problem and a primal-dual algorithm also approximates it within a factor of $(t+1)$ \cite{hochbaum1982approximation}. \texttt{Min-$t$-Claw-Del} is $\NP$-complete even for bipartite graphs by a generic result of Yannakakis \cite{YannakakisM}. For bipartite graphs, \texttt{Min-$t$-Claw-Del} can be approximated within a factor of $t$ by an iterative rounding algorithm \cite{kumar2014approximation}. This approximation factor is improved to $O(\log(t+1))$ \cite{guruswami2017inapproximability}, because minimum $t$-claw transversal is approximable within $O(\log(t+1))$ and for bipartite graphs this problem coincides with \texttt{Min-$t$-Claw-Del}. 
 
A graph $G=(V, E)$ is called a split graph if the vertex set can be partitioned into $A$ and $B$ such that $G[A]$ is a complete graph and $G[B]$ is an independent set. Braberman et.al. \cite{bonomo2020linear} proved that \texttt{Min-Claw-Del} is $\NP$-complete for split graphs and hard to approximate better than 2, assuming $\UGC$. Later Hsieh et.al. \cite{hsieh2021d} proved that, for $t \geq 3$, \texttt{Min-$t$-Claw-Del} is $\NP$-complete for split graphs without having $(t+1)$-claws and split graphs with diameter 2. They also proposed a polynomial time algorithm for $t$-block graphs.

Fujito \cite{fujitoM} studied the node deletion problems for several graph properties, notably node deletion problems for matroidal properties. A graph property $\pi$ is a matroidal property if, on any graph $G$, the edge sets of subgraphs of $G$ satisfying $\pi$ form a family of independent sets in some matroid defined on the edge set of $G$. These kinds of node deletion problems are formulated as a submodular set cover problem and a primal-dual approximation algorithm is used to compute an approximate deletion set. For some problems, it has been proved that they are approximable within a constant factor.  For example, if $\pi$ is uniformly $k$-sparse then the corresponding node deletion problem is approximable within a factor of 2 \cite{fujitoM}. However, even if the property $\pi$ is not matroidal the corresponding node deletion problem can be approximated within a nice 
bound under certain assumptions. For example, a bounded degree deletion problem is formulated as a submodular set cover problem with help of a 2-polymatroid matching on the edge set of the input graph \cite{fujito2017approximating}. It is important to observe that this graph property ($b$-matching) is not matroidal. 

It is easy to observe that in a split graph every $t$-claw has its center vertex in its clique partition. Based on this observation, we consider the complexity of minimum one-sided $t$-claw deletion problem (\texttt{Min-$t$-OSBCD}) for bipartite graphs. Given  a bipartite graph $G=(A \cup B, E)$, in \texttt{Min-$t$-OSBCD} it is required to find a vertex set $S$ such that $G[(A\cup B) \setminus S]$ has no $t$-claw with the center vertex in $A$. We prove that \texttt{Min-$t$-OSBCD} is $\UGC$-hard to approximate better than $t$ and can be approximated within a factor of $t$. This lower bound result on the approximability of \texttt{Min-$t$-OSBCD} implies that  \texttt{Min-$t$-Claw-Del} for the split graph can not be approximated within a factor smaller than $t$ unless $\UGC$ is false. This lower bound is a stronger lower bound on \texttt{Min-$t$-Claw-Del} for split graphs than that of Braberman et.al. \cite{bonomo2020linear}. The graph property ``one-sided-$t$-claw free on bipartite graphs'' is matroidal. The vertex deletion problem associated with this graph property can be formulated as a submodular set cover problem via a 2-polymatroid matching function on the edge set of the input graph. We prove that if $d(v) \geq 2(t-1)$, for all $v \in A$, in the input graph $G=(A \cup B, E)$, then \texttt{Min-$t$-OSBCD} is $\APX$-complete and can be approximated within a factor of 2. A similar result holds for \texttt{Min-$t$-Claw-Del} for split graphs.

\section{Preliminaries}
Given a vertex set $X \subseteq (A \cup B)$ in a bipartite graph $G=(A \cup B, E)$, we define $X_A =X \cap A$ and $X_B = X \cap B$. For a vertex $v$ in $G$, $\delta(v)$ denotes the set of edges incident on $v$. For an edge set $F \subset E$ and a vertex $v$ in $G$, $d_F(v)$ denoted the degree of the vertex $v$ in the subgraph $(A \cup B, F)$. For a set $X$ of vertices in $G$, we define $\delta(X)$ as the set of edges in $G$ which are incident on a vertex in $X$. Given an edge set $F \subseteq E$ of a graph $G=(V, E)$, we define $d_F(v)$ as the degree of vertex $v$ in $(V, F)$.

A vertex set $S \subseteq V$ is a vertex cover in $G=(V, E)$ if $S \cap \{u, v\} \neq \emptyset$. In minimum vertex cover (\texttt{Min-VC}), given a graph $G$ it is required to find a vertex cover $S$ in $G$ of minimum size. A hypergraph $G=(V, E)$ is consisting of a finite vertex set $V$ and a finite edge set $E$, where each hyperedge $e \in E$ is a subset of $E$. Given a positive integer $t \geq 2$, a hypergraph $G$ is called a $t$-uniform hypergraph if each hyperedge $e$ of $G$ is a $t$-element set. A vertex cover in a hypergraph $G=(V, E)$ is a vertex set $S \subseteq V$ such that $S \cap e \neq \emptyset$. In minimum $t$-hypergraph vertex cover problem (\texttt{Min-$t$-Hyper-VC}), given a $t$-uniform hypergraph, $G$ it is required to find a vertex cover $S$ in $G$ of minimum size. \texttt{Min-VC} is known to be $\APX$-complete and approximable within a factor of $2$. This constant approximation factor for \texttt{Min-VC} is the best possible assuming $\UGC$ \cite{khot_regev}. Also, \texttt{Min-$t$-Hyper-VC} is hard to approximate within a factor of better than $t$, assuming $\UGC$ \cite{khot_regev}.

For a finite set $N$, a non-decreasing, submodular, and integer-valued function $f$ defined on $2^N$ with $f(\emptyset)=0$ is called a polymatroid function and $(N, f)$ a polymatroid. If $f$ additionally satisfies $f(\{j\}) \leq k$, for each $j \in N$, then $(N, f)$ is called a $k$-polymatroid. For any polymatroid $(N, f)$ define another set function $f^d$ such that $f^d(S) = \sum_{j \in S}f(\{j\}) - (f(N) - f(N\setminus S))$. $f^d$ is a polymatroid function and $(N, f^d)$ is called the dual polymatroid of $(N, f)$. In a 2-polymatroid $(E, f)$, a subset $F \subset E$ is a matching in $(E, f)$ if $f(F) = 2|F|$. Also, a subset $F \subset E$ is spanning in $(E, f)$ if $f(E) = f(F)$.

\begin{proposition} \cite{fujito2017approximating}\\
(a) The dual $(E, f^d)$ of a 2-polymatroid $(E, f)$ is also a also a 2-polymatroid. \\
(b) A subset $F \subseteq E$ is a matching in the 2-polymatroid $(E, f)$ if and only if $E\setminus F$ is spanning in $(E, f^d)$. 
\end{proposition}

\section{Hardness results}
In this section, we obtain a lower bound on the approximability of \texttt{Min-$t$-OSBCD} that matches with the approximation upper bound. By using this lower bound, we obtain a similar lower bound on the approximability of \texttt{Min-$t$-Claw-Del} for split graphs.

\begin{theorem} \label{thm1}
For a fixed integer $t \geq 3$, assuming $\UGC$, there is no polynomial time approximation algorithm for \texttt{Min-$t$-OSBCD} within a factor better than $t$. Such a lower bound exists for bipartite graphs $G=(A \cup B, E)$ with $d(v) = t$, for each vertex $v \in A$.
\end{theorem}
\begin{proof}
Let $t \geq 3$ be a fixed integer. We prove this theorem by establishing a reduction from \texttt{Min-$t$-Hyper-VC}. Let $G=(V, E)$ be an instance of \texttt{Min-$t$-Hyper-VC}. We assume that $G$ is a $t$-uniform hypergraph and for any hyperedge, $e \in E$ there exists a hyperedge $e' \in E$ such that $e$ and $e'$ are disjoint. We also assume that $V=\{v_1, v_2, \ldots,v_n\}$ and $E=\{e_1, e_2, \ldots, e_m\}$. From $G$ we construct an instance $G'=(V', E')$ of \texttt{Min-$t$-OSBCD} as follows. 
 \begin{itemize} 
 \item[.] For each vertex $v \in V$, we construct a vertex gadget 
consisting of a vertex $v$ in $G'$.
 \item[.] For each edge $e \in E$, we construct an edge gadget $H_e$ consisting of $n$ distinct vertices $e^1, e^2, \ldots, e^n$.
 \item[.] Finally, we connect vertex gadgets with edge gadgets by introducing the following edges. For each vertex $v \in e$, we introduce $n$ edges $\{(v, e^i) \mid 1 \leq i \leq n\}$. 
\end{itemize}
It can be observed that $G'$ is a bipartite graph with vertex bipartition $V' = A \cup B$, where $A = \cup_{e \in E}\{e^1, e^2, \ldots, e^n\}$ and $B = V$.

\begin{claim}
$S \subseteq V$ is a minimal vertex cover in $G$ if and only if $S$ is a minimal vertex deletion set for \texttt{Min-$t$-OSBCD} for the instance $G'$.
\end{claim}
\begin{proof}
It is important to observe that each vertex in $A$ has degree $t$ and they are the center vertices of all the one-sided $t$-claws in $G'$. 

If $e=\{u_1, \ldots, u_t\} \in  E$ then each $e^i$ ($i=1, 2, \ldots, n$) is adjacent to all $u_i \in e$. This implies that if $S$ is a minimal vertex cover in $G$ then $S$ is a minimal $t$-\texttt{OSBCD} set in $G'$. This also implies that $|S_{o}| \leq |V|$, where $S_o$ is a minimum size $t$-\texttt{OSBCD} set in $G'$.

Conversely, let $S$ be a minimal $t$-\texttt{OSBCD} set in $G'$. If $S \nsubseteq V$, then we claim that $|S| > |V|$. If we assume that $S \nsubseteq V$ then $S$ must contain at least one vertex from $\{e^1, e^2, \ldots, e^n\}$, for some $e \in E$. Without loss of generality, assume that $e^i \in S$, for some $i \in \{1, 2, \ldots, n\}$. Since $S$ is a minimal $t$-\texttt{OSBCD} and $e^i \in S$, all its $t$ neighbours $u_1, \ldots, u_t$ are not in $S$. As $u_1, \ldots, u_t$ are not in $S$, $S$ must contain the vertices $e^1, e^2, \ldots, e^n$. $|S| > |V|$ as we have assumed that $E$ has disjoint pairs of hyperedges.  

Since, a minimum vertex cover $S_o$ in $G$ is of size at most $n$, without loss of generality, we assume that any minimal $t$-\texttt{OSBCD} set $S$ in $G'$ is a subset of $V$. Then, from the construction of $G'$, it follows that any minimal $t$-\texttt{OSBCD} set $S$ in $G'$ which is a subset of $V$ must be a minimal vertex cover in $G$. 
\end{proof}

From the above claim, it follows that the size of a minimum vertex cover in $G$ and the size of a minimum $t$-\texttt{OSBCD} in $G'$ are equal. Therefore \texttt{Min-$t$-OSBCD} is hard to approximate within a factor smaller than $t$ as it is $\UGC$-hard to approximate \texttt{Min-$t$-Hyper-VC} within a factor smaller than $t$ \cite{khotregev}.  
\end{proof}

By using this lower bound result on \texttt{Min-$t$-OSBCD}, we obtain a similar lower bound result for \texttt{Min-$t$-Claw-Del} for split graphs. This improves the lower bound result for \texttt{Min-$t$-Claw-Del} to $t$, for split graphs, and it matches with the upper bound on the approximation factor.

\begin{theorem} \label{thm2}
Assuming $\UGC$, for split graphs \texttt{Min-$t$-Claw-Del} can not be approximated better than $t$, where $t \geq 3$ is a constant integer.
\end{theorem}
\begin{proof}
Given an instance $G=(A \cup B, E)$ of \texttt{Min-$t$-OSBCD}, we construct a split graph $H=(A \cup B, E')$, an instance of \texttt{Min-$t$-Claw-Del}, with $E' = E \cup \{(u, v) \mid u, v \in A\}$. It is easy to observe that $S$ is a one-sided $t$-claw deletion set in $G$ if and only if $S$ is a split-$t$-claw deletion set in $H$. From this observation and \autoref{thm1}, it follows that \texttt{Min-$t$-Claw-Del} is hard to approximate with a factor better than $t$ unless $\UGC$ is true.  
\end{proof}

Next, we prove that \texttt{Min-$t$-OSBCD} is $\APX$-complete for bipartite graphs $G=(A \cup B, E)$ with $d(v) \geq 2(t-1)$, for all $v \in A$. 

\begin{theorem} \label{thm3}
For $t \geq 3$, \texttt{Min-$t$-OSBCD} is $\APX$-complete for bipartite graphs $G=(A \cup B, E)$ with $d(v) \geq 2(t-1)$, for all $v \in A$.
\end{theorem}
\begin{proof}
It is known that for $t \geq 3$, \texttt{Min-$t$-OSBCD} is approximable within a factor of $t$. Hence it remains to prove that it is $\APX$-hard. We prove this by an $L$-reduction \cite{papaYa} from Minimum Vertex Cover for $t$ regular graphs. Given a graph $G=(V, E)$, an instance of minimum vertex cover, in polynomial time we construct bipartite graph $H=(A \cup B, E')$, an instance of \texttt{Min-$t$-OSBCD} as follows. Here we assume that $|V| = n$ is sufficiently larger than $t$.
 \begin{itemize}
  \item[.] We introduce two vertex sets $V$ and $E$ in $H$. For each edge $e=(u, v)$ introduce the edges $(e, u)$ and $(e, v)$.
  \item[.] We make $x=\lfloor \frac{t}{2} \rfloor -1$ copies $V^1, V^2, \ldots, V^x$ of $V$. For each edge $e$ in $G$, we introduce the edges $(e, u^i), (e, v^i)$, for $1 \leq i \leq x$. Here $u^i \in V^i$ is the $i$th copy of $u \in V$.
  \item[.] We make $2n-1$ number of copies $E^1, E^2, \ldots, E^{2n-1}$ of $E$ and introduce the edges between these added sets of vertices to the vertices in $V, V^1, V^2, \ldots, V^x$ as follows. If $e^i \in E^i$ then join this vertex to all the neighbors of $e \in E$ in $V \cup V^1 \ldots \cup V^x$.
  \item[.] Finally, we make a new set $P$ having $t-2$ or $t-1$ vertices depending on whether $t$ is even or odd, respectively. We make each vertex in $P$ adjacent to every vertex in $E \cup E^1 \cup \ldots E^{2n-1}$.
\end{itemize}
\begin{figure}[h]
  \begin{center} 
    \includegraphics[scale=.4]{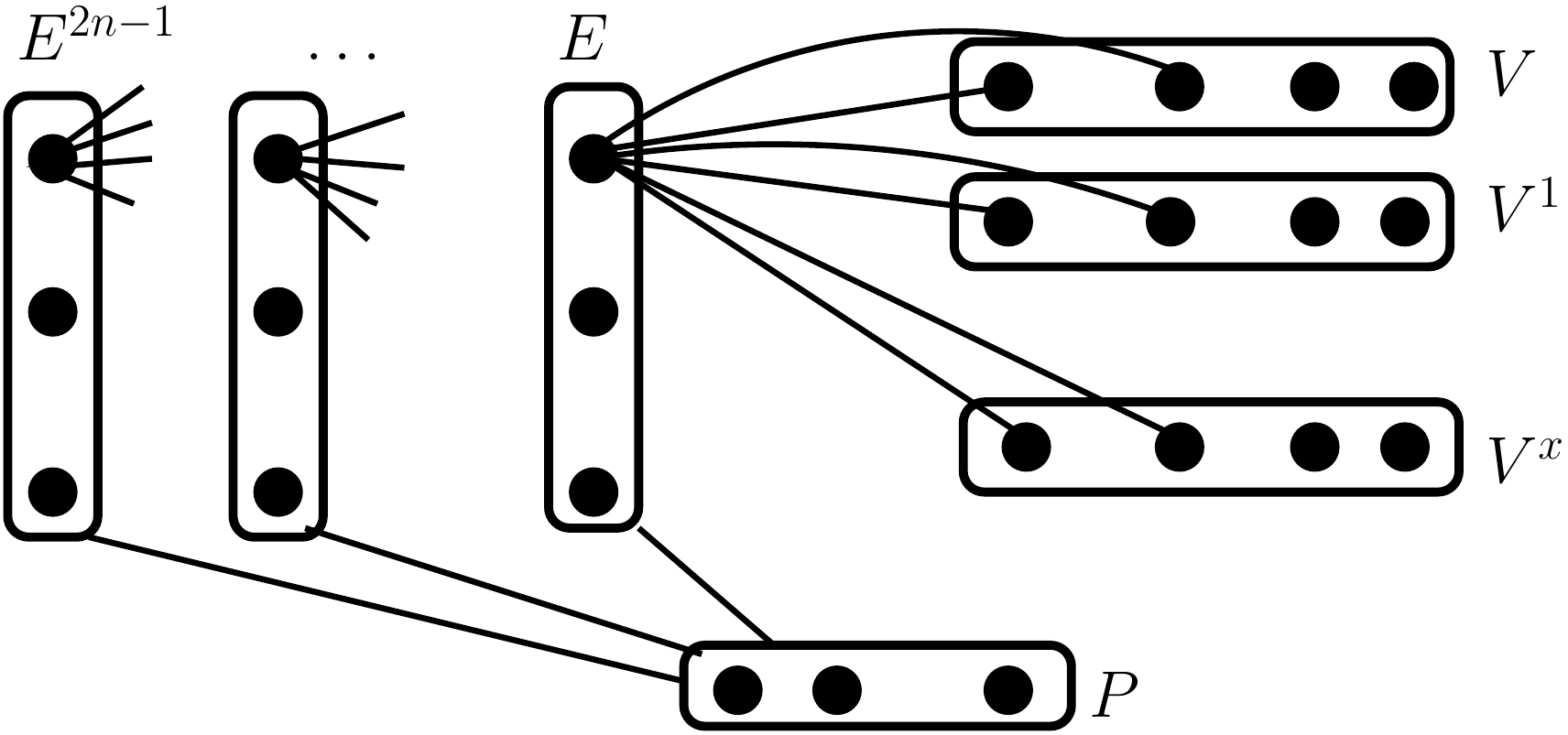}
      \end{center}
 \caption{An illustration of $H$.}
\label{fig1}
\end{figure}
It is easy to observe that $H$ is a bipartite graph with vertex bipartition as $A = E \cup E^1 \cup \ldots, E^{2n-1}$ and $B = V \cup V^1 \cup \ldots V^x\cup P$. In the graph $H$, it is important to observe that degree of each vertex in $A$ is $2(t-1)$. We also assume that the size of a minimum vertex cover $S_o$ in $G$ is larger than $|P|$. 
 
Now, we claim that size of a minimum vertex cover in $G$ is $k$ if and only if the size of a minimum $t$-\texttt{OSBCD} in $H$ is $|P| + k$.
 
Given a minimal vertex cover $S$ in $G$, we define $S' = S \cup P$. $S'$ is a $t$-\texttt{OSBCD} deletion set in $H$, because after deletion of vertex set $S'$ from $H$, degree of each vertex in $A$ are of at most $t-1$. Here, $|S'| = |P| + |S|$.

Let $S'$ be a minimal $t$-\texttt{OSBCD} set in $H$. If $S'$ contains a vertex $y$ from $A$ then all its $2n-1$ copies must be in $S'$, because $y$ has at least $t$ neighbors, not in $S'$. This would imply that $|S'| \geq 2n$ and hence $S'$ can not be a minimum $t$-\texttt{OSBCD} deletion set in $H$. Therefore, we shall assume that a minimal $t$-\texttt{OSBCD} set $S'$ in $H$ is a subset of $B$. With a similar argument, it can be observed that $S'$ contains $T$ and a subset of $V$ (because $n$ is sufficiently large than $t$). This would imply that $S=S' \cap V$ is a minimal vertex cover in $G$ and $|S'| = |P| + |S|$. 

From these observations, it follows that if $S'_o$ is a minimum $t$-\texttt{OSBCD} deletion set in $H$ then $|S'_o| = |P| +|S_o| \leq 2 |S_o|$. Also, from any minimal $t$-\texttt{OSBCD} deletion set $S'$ in $H$ we construct a minimal vertex cover $S$ in $G$ such that $|S'|= |P| + |S|$. This implies that it is a $L$-reduction with $\alpha =2$ and $\beta = 1$.  
\end{proof}

By using the reduction in the proof of Theorem \ref{thm2}, we have the following result.
\begin{theorem}
For split graphs, \texttt{Min-$t$-Claw-Del} is $\APX$-complete even when each vertex $v \in A$ has at least $2(t-1)$ neighbours in $B$.
\end{theorem}

\section{Approximation Algorithms}
In this section, we define a matroid $M_t=(E, {\mathcal F})$ on a bipartite graph $G=(A \cup B, E)$ associated with one-sided $t$-claws in $G$. Given a positive integer $t \geq 3$, we define ${\mathcal F}$ as the collection of edge sets $F\subseteq E$ such that the subgraph $(A\cup B, F)$ has no $K_{1, t}$ with the center vertex in $A$. In other words, if $F$ is an independent set in the matroid $M_t$ then the subgraph $(A\cup B, F)$ of $G$ can have a vertex in $B$ with degree larger than $(t-1)$, but no vertex in $A$ can have degree larger than $(t-1)$. Clearly, $M_t$ is a hereditary system. Next, we  show that for any pair of independent sets $F_1$ and $F_2$ in $M_t$, if $|F_1| < |F_2|$ then there exists an element $e \in F_2\setminus F_1$ such that $F_1 \cup \{e\}$ is an independent set in $M_t$. Since $(A\cup B, F_1)$ and $(A\cup B, F_2)$ are bipartite graphs and $|F_1| < |F_2|$, there must exist a vertex $v \in A$ such that $d_{F_1}(v) < d_{F_2}(v)$. Now, it is easy to observe that $F_1 \cup \{e\}$ is an independent set, for each edge $e \in F_2 \setminus F_1$ incident on $v$. Hence, $M_t$ is a matroid.

Given a bipartite graph $G=(A\cup B, E)$ and a positive inetger $t \geq 3$, we define a function $f_t : 2^E \rightarrow Z_+$ such that $$f_t(F) = 2\left[\sum_{v\in A}\mbox{min}\{(t-1), d_F(v)\}\right].$$ 
We will show that the function $f_t$ corresponds to the 2-polymatroid function corresponding to \texttt{Min-$t$-OSBCD}. 

\begin{lemma}
 For any bipartite graph $G=(A \cup B, E)$, $(E, f_t)$ is a 2-polymatroid. Also, $F \subseteq E$ is a matching in $(E, f_t)$ if and only if $(A\cup B, f)$ is one-sided $t$-claw free.
\end{lemma}
\begin{proof}
 It is easy to observe that $f_t$ is a non-decreasing, integer-valued function with $f_t(\emptyset)=0$. Now, it remains to show that $f_t$ is a submodular function and we will prove that $$f_t(X \cup \{e_1\}) + f_t(X \cup \{e_2\}) \geq f_t(X \cup \{e_1, e_2\}) + f_t(X),$$ for every $X \subseteq E$ and $e_1, e_2 \in (E \setminus X)$. 
 
Let $X$ be any subset of $E$ and $e_1, e_2$ be two edges in $E\setminus X$. If both $e_1$ and $e_2$ are incident on a vertex $v \in A$ with $d_X(v) \leq t-3$ then $f_t$ satisfies the above inequality as equality. If $d_X(v) = t-2$, then it is satisfied with strict inequality. If $d_X(v) \geq t-1$ then it is satisfied as equality.

Let us assume that $e_1$ is incident on $u \in A$ and $e_2$ is incident on $v \in A$. If $d_X(u)$ and $d_X(v)$ are less than $t-1$, then the above inequality is satisfied as an equality, otherwise, it is satisfied as a strict inequality.

The second part of the proposition follows directly from the definition of matching and $f_t$.  
\end{proof}

We assume that $G$ has no vertex $v \in A$ with $d(v) \leq (t-1)$. This is because such a vertex does not generate a $t$-claw with a center vertex in $A$, and the edges incident on $v$ has no impact on its neighbors in $B$.
\begin{lemma}
 Let $t \geq 3$ be a fixed positive integer and $G=(A \cup B, E)$ be a bipartite graph with $d(v) \geq t-1$ for all $v \in A$. Then 
 \begin{itemize}
\item[(a)] $f_t(E) = 2[\sum_{v\in A}\mbox{min}\{(t-1), d_E(v)\}]= 2|A|(t-1).$
\item[(b)] $f_t(e) = 2$, for each $e \in E$.
\item[(c)] $f_t^d(E) = \sum_{e \in E}f_t(e) - (f_t(E) - f_t(\emptyset))= 2|E| - f_t(E) = 2|E| - 2|A|(t-1) \\ = 2\sum_{v \in A}[d(v) - t +1]$.
\item[(d)] For any $F \subseteq E$,  
\begin{eqnarray*} 
f_t^d(F) & = & \sum_{e \in F}f_t(e) - (f_t(E) - f_t(E \setminus F)) \\
         & = & 2|F| - 2(\sum_{v \in A}\mbox{min}\{t-1, d(v)\} - \sum_{v \in A}\mbox{min}\{t-1, d_{E\setminus F}(v)\}) \\
         & = & 2\sum_{v \in A}d_F(v) - 2\sum_{v \in A}\mbox{max}\{0, \mbox{min}
         \{t-1, d(v)\} - d_{E \setminus F}(v)\} \\
         & = & 2 \sum_{v \in A}\mbox{min}\{d_F(v), d(v) - t+1\}.
\end{eqnarray*}
 \item[(e)] If $v \in B$ then $N(v) \subseteq A$ and
 \begin{eqnarray*}
  f^d_t(\delta(v)) & = & \sum_{u \in \delta(v)}f_t(\{u\}) - [f_t(E) - f_t(E \setminus \delta(v))] \\
                   & = & 2\sum_{u \in A}\mbox{min}\{d_{\delta(v)}(u), d(u) -t +1\} \\
                   & = & 2d(v).
 \end{eqnarray*}
Last step follows from the fact that $d_{\delta(v)}(u) = 1$ if $u$ is a neighbour of $v$ and $d_{\delta(v)}(u) = 0$ if $u \in A$ is not a neighbour of $v$.

\item[(f)] If $v \in A$ then $N(v) \subseteq B$ and 
  \begin{eqnarray*}
  f^d_t(\delta(v)) & = & 2[d(v) -t +1].
 \end{eqnarray*}
\end{itemize}
\end{lemma}

Based on these properties, \texttt{Min-$t$-OSBCD} can be formulated as a submodular set cover problem. By following the submodular set cover formulation of a node deletion problem for a graph property satisfying matroidal property \cite{fujitoM}, \texttt{Min-$t$-OSBCD} can be formulated as given below. Here, we consider the vertex weighted version of \texttt{Min-$t$-OSBCD}.
\begin{eqnarray*}
 (P):  Minimize & \sum_{v \in V} w(v) x_v  & ~~~\\
             Subject~ to & \sum_{v \in S} f_t^d(\delta_S(v))x_v \geq f_t^d(E[S]),  & \forall S \subseteq V \\
                        & x_v \in \{0, 1\},  & \forall v \in V. 
\end{eqnarray*}

The dual $(D)$ of a linear program relaxation of $(P)$ can be written as follows.
\begin{eqnarray*}
 (D):  Maximize & \sum_{S \subseteq V} f_t^d(E[S])y_S  & ~~~\\
             Subject~ to & \sum_{S:v \in S} f_t^d(\delta_S(v))y_S \leq w_v,  & \forall v \in V \\
                        & y_S \geq 0,  & \forall S \subseteq V. 
\end{eqnarray*}
Based on these primal-dual formulations of \texttt{Min-$t$-OSBCD} a primal-dual approximation algorithm can be designed and it is described in Algorithm-\ref{alg1}.

\begin{algorithm2e}[H] \label{alg1}
\KwIn{A bipartite graph $G=(A \cup B,E)$ and $w:A\cup B \rightarrow \mathbb{Q}^+$;}
\KwOut{A $t$-\texttt{OSBCD} deletion set $F$;}
\vspace*{.1cm}
$F' = \emptyset$; $S = V$; $y= 0$\;
\While{$F'$ is not a $t$-\texttt{OSBCD} set in $G$}{
  Increase $y_S$ until, for some $v \in S$, the dual constraint in $(D)$ for $v$ becomes tight\;
  $F' = F' \cup \{v\}$\;
  $S = S \setminus \{v\}$\;
}
Compute a minimal $t$-\texttt{OSBCD} set $F \subseteq F'$ by using reverse deletion method on $F'$\;
return $(F)$\;
\caption{{Approximation Algorithm for \texttt{Min-$t$-OSBCD}}}
\end{algorithm2e}

The performance of this primal-dual algorithm is described in the following Lemma.
\begin{lemma}
 The performance ratio of Algorithm-\ref{alg1} is bounded above by 
 $${\displaystyle \theta = \mbox{max}\{\frac{\sum_{v \in S}f_t^d(\delta(v))}{f_t^d(E)} \mid S \mbox{~is a minimal~} t\mbox{-OSBCD set in~} G\}}.$$
\end{lemma}

Next, we prove that $\theta = 2$ under some assumption on the number of edges in $G$.

\begin{lemma}
 Let $t \geq 3$ be a positive integer and $G=(A \cup B, E)$ be a bipartite graph with $d(v) \geq 2(t-1)$, for all $v \in A$. For any minimal $t$-\texttt{OSBCD} set $X$ in $G$, $$2f_t^d(E) \geq \sum_{v \in X}f_t^d(\delta(v)).$$
\end{lemma}
\begin{proof}
We will prove that $2f_t^d(E) - \sum_{v \in X}f_t^d(\delta(v)) \geq 0.$ 
\begin{eqnarray*}
 & & 2f_t^d(E) - \sum_{v \in X}f_t^d(\delta(v)) \geq 0 \\
 & \Leftrightarrow & 4|E| - 4(t-1)|A| -2 \sum_{v \in X_A}d(v) + 2(t-1)|X_A| - 2\sum_{v \in X_B}d(v) \geq 0 \\
 & \Leftrightarrow & \sum_{v \in X_A}[d(v) - t +1] + \sum_{v \in X_B}d(v) \leq 2[|E| -(t-1)|A|]
\end{eqnarray*}
We can write $\sum_{v \in X_B}d(v) = |(X_A, X_B)| + \sum_{v \in Y_A}[d(v) -t +1] + |(Y_A, X_B)| - \sum_{v \in Y_A}[d(v) -t +1].$ By using this relation in the above inequality, we have
\begin{eqnarray*}
 & & |(X_A, X_B)| + |(Y_A, X_B)| - \sum_{v \in Y_A}[d(v) - t +1] \\ & & \hspace*{4cm} \leq \sum_{v \in X_A}[d(v) - t +1] + \sum_{v \in Y_A}[d(v) - t +1] \\
 & \Leftrightarrow & |(X_A, X_B)| - \sum_{v \in X_A}[d(v) - t +1]  + |(Y_A, X_B)| \leq 2 \sum_{v \in Y_A}[d(v) - t +1].
\end{eqnarray*}

Since $X$ is a minimal deletion set $|(X_A, X_B)| < \sum_{v \in X_A}[d(v) - t +1]$. Therefore, it remains to prove $$|(Y_A, X_B)| \leq 2 \sum_{v \in Y_A}[d(v) - t +1].$$

Now, $d(v) \leq 2(d(v) -t + 1)$ as $d(v) \geq 2(t-1)$. This implies that $|(Y_A, X_B)| \leq \sum_{v \in Y_A}d(v) \leq 2\sum_{v \in Y_A}(d(v) -t+1)$.  
\end{proof}

From this Lemma, we have the following result.

\begin{theorem} \label{thm5}
 Let $t \geq 3$ be a positive integer and $G=(A \cup B, E)$ be a bipartite graph with $d(v) \geq 2(t-1)$, for all $v \in A$. Then Algorithm-\ref{alg1} approximates \texttt{Min-$t$-OSBCD} within a factor of 2.
\end{theorem}

It is easy to observe that \texttt{Min-$t$-Claw-Del} for split graphs is $L$-reduciable to \texttt{Min-$t$-OSBCD} with $\alpha = 1$ and $\beta = 1$. Therefore, we have the following result.

\begin{theorem} \label{thm6}
\texttt{Min-$t$-Claw-Del} can be approximated within a factor of 2, where the input instance is a split graph $G=(A \cup B, E)$ with each vertex $v \in A$ having at least $2(t-1)$ neighbors in $B$.
\end{theorem}

Next, we are considering two maximization problems \texttt{Max-$t$-Claw-Subgraph} and \texttt{Max-$t$-OSBC-Subgraph} which are  complementary problems of \texttt{Min-$t$-Claw-Del} and \texttt{Min-$t$-OSBCD}, respectively. Given a split graph $G$, in \texttt{Max-$t$-Claw-Subgraph} it is required to find a vertex set $F$ in $G$ of maximum size such that $G[F]$ is a subgraph of $G$ without  having a $t$-claw. Similarly, for bipartite graphs \texttt{Max-$t$-OSBC-Subgraph} for is defined. The $L$-reduction in the proof of Theorem \ref{thm3} can be seen as a $L$-reduction from maximum independent set to \texttt{Max-$t$-OSBC-Subgraph}. The reduction in the proof of Theorem \ref{thm2} can be modified to a $L$-reduction from \texttt{Max-$t$-OSBC-Subgraph} to \texttt{Max-$t$-Claw-Subgraph}. Therefore, it follows that these two maximization problems are $\APX$-hard. Now, we will prove that they can be approximated within a factor of $2-\frac{1}{t}$.

\begin{theorem}
\texttt{Max-$t$-OSBC-Subgraph} can be approximated within a factor of $2-\frac{1}{t}$. Also, for split graphs, \texttt{Max-$t$-Claw-Subgraph} has a $2-\frac{1}{t}$ factor approximation algorithm.
\end{theorem}
\begin{proof}
Here, we assume that the input instance for \texttt{Max-$t$-OSBC-Subgraph} is a vertex weighted bipartite graph $G=(A \cup B, E)$. By using Algorithm \ref{alg1}, we compute a $t$-\texttt{OSBCD} set $S$ in $G$. Then, we compute the maximum weight set among $V \setminus S, A$ and $B$. Let $F$ be this set. It is easy to observe that $G[F]$ has no $t$-claw with center vertex in $A \cap F$. Also $w(V) \leq 2w(F)$. Without loss of generality, assume that $w(A) \geq w(B)$. Then $w(V) \leq 2w(A)$. This implies that $w(V) \leq 2w(F)$ as $w(A) \leq w(F)$.

Let $S_0$ be a minimum $t$-\texttt{OSBCD} set in $G$. Then $F_0=V \setminus S_0$ induces a maximum $t$-claw free subgraph in $G$. Since $S$ is a $t$-approximate solution of \texttt{Min-$t$-OSBCD}, we have
\begin{eqnarray*}
   & & w(S)  \leq  tw(S_0) \\
 & \Rightarrow &  w(V) - w(V\setminus S) \leq t[w(V) - w(F_0)] \\
  & \Rightarrow  &  w(V) - w(F) \leq t[w(V) - w(F_0)] ~~(as~ w(F) \geq w(V\setminus S))\\
 & \Rightarrow & tw(F_0) \leq (t-1)w(V) + w(F) \leq (2t-1) w(F) \\
 & \Rightarrow & \frac{w(F_0)}{w(F)} \leq 2-\frac{1}{t}.
\end{eqnarray*}
This implies that \texttt{Max-$t$-OSBC-Subgraph} has a $2-\frac{1}{t}$ factor approximation algorithm.

By using a similar argument, it can be proved that \texttt{Max-$t$-Claw-SUbgraph}, for split graphs, has a $2-\frac{1}{t}$ factor approximation algorithm.
 
\end{proof}

Again by using the same arguments and Theorems \ref{thm5}, \ref{thm6}, we can prove the following result.
\begin{theorem}
Let $t \geq 3$ be a positive integer and $G=(A \cup B, E)$ be a bipartite graph with $d(v) \geq 2(t-1)$, for all $v \in A$. \texttt{Max-$t$-OSBC-Subgraph} can be approximated within a factor of $\frac{3}{2}$ for such kind of bipartite graphs. Also, \texttt{Max-$t$-Claw-Subgrah} can be approximated within a factor of $\frac{3}{2}$, for split graphs $G=(A \cup B, E)$ with each vertex $v \in A$ has at least $2(t-1)$ neighbours in $B$.
\end{theorem}

\bibliographystyle{abbrv}

\end{document}